\documentclass[12pt]{article}

\usepackage{fullpage}

\usepackage{xspace}             % \xspace
\usepackage{hyperref}

\usepackage{amsmath}
\usepackage{amsthm}
\usepackage{amsfonts}
\usepackage{amssymb}

\usepackage{caption}
\usepackage{subcaption}
\usepackage{xcolor,cancel}

\usepackage{graphicx}           % \includegraphics 
\usepackage{import}             % \subimport
\usepackage{epstopdf}           % for included gnuplots

\usepackage{enumerate}
\usepackage{times}
\usepackage{algorithm}
\usepackage{algpseudocode}

\algnewcommand\algorithmicinput{\textbf{Input:}}
\algnewcommand\Input{\item[\algorithmicinput]}
\algnewcommand\algorithmicoutput{\textbf{Output:}}
\algnewcommand\Output{\item[\algorithmicoutput]}

\newtheorem{theorem}{Theorem}
\newtheorem{lemma}{Lemma}
\newtheorem{corollary}{Corollary}
\newtheorem{definition}{Definition}

\newcommand{\eps}{\epsilon}

\newcommand{\diag}{\mathtt{diag}}

\newcommand{\ali}[1]{\|a_{., #1}\|_1}
\newcommand{\alo}[1]{\|a_{#1, .}\|_1}

\newcommand{\x}{\mathbf{x}}
\newcommand{\xt}[1]{\mathbf{x^{(#1)}}}

\newcommand{\fb}{f^{(\mathcal{B})}}
\newcommand{\fbk}[1]{f^{(\mathcal{B}_{#1})}}

\newcommand{\gbk}[1]{G^{(\mathcal{B}_{#1})}}
\newcommand{\B}{\mathcal{B}}
\DeclareMathOperator*{\argmax}{arg\,max}

\title{Strictly Balancing Matrices in Polynomial Time Using Osborne's Iteration}

\author{Rafail Ostrovsky\thanks{Research supported in part by NSF grants 1065276, 1118126 and 1136174, US-Israel BSF grants, OKAWA Foundation Research Award, IBM Faculty Research Award, Xerox Faculty Research Award, B. John Garrick Foundation
Award, Teradata Research Award, and Lockheed-Martin Corporation Research Award.
This material is also based upon work supported in part by DARPA
Safeware program. The views expressed are those of the authors and do not reflect the official policy or position of the Department of Defense or the U.S. Government.} \\
\and
Yuval Rabani\thanks{Research supported in part by ISF grant 956-15, by BSF grant 
2012333, and by I-CORE Algo.} \\
\and
Arman Yousefi\footnotemark[1]}

\begin{document}

\maketitle
\begin{abstract}
Osborne's iteration is a method for balancing $n\times n$ matrices 
which is widely used in linear algebra packages, as balancing preserves
eigenvalues and stabilizes their numeral computation. The iteration
can be implemented in any norm over $\mathbb{R}^n$, but it is normally 
used in the $L_2$ norm. The choice of norm not only affects the desired
balance condition, but also defines the iterated balancing step itself.

In this paper we focus on Osborne's iteration in any $L_p$ norm, where 
$p < \infty$. We design a specific implementation of Osborne's iteration 
in any $L_p$ norm that converges to a strictly $\eps$-balanced matrix in 
$\tilde{O}(\eps^{-2}n^{9} K)$ iterations, where $K$ measures, roughly, 
the {\em number of bits} required to represent the entries of the input 
matrix.

This is the first result that proves that Osborne's iteration in the $L_2$
norm (or any $L_p$ norm, $p < \infty$) strictly balances matrices in
polynomial time. This is a substantial improvement over our recent 
result (in SODA 2017) that showed weak balancing in $L_p$ 
norms. Previously, Schulman and Sinclair (STOC 2015) showed 
strong balancing of Osborne's iteration in the $L_\infty$ norm. Their 
result does not imply any bounds on strict balancing in other norms. 
\end{abstract}

\thispagestyle{empty}
\newpage
\setcounter{page}{1}

%%%%%%%%%%%%%%%%%%%%%%%%%%%%%%%%%%%%
%%%%%   Introduction
%%%%%%%%%%%%%%%%%%%%%%%%%%%%%%%%%%%%
\section{Introduction}

\paragraph{Problem statement and motivation.}
This paper analyzes the convergence properties of Osborne's
celebrated iteration~\cite{Osborne} for balancing matrices. Given
a norm $\|\cdot\|$ in $\mathbb{R}^n$, an $n\times n$ matrix
$A$ is balanced if and only if for all $i$, the $i$-th row of $A$
and the $i$-th column of $A$ have the same norm. The problem
of balancing a matrix $A$ is to compute a diagonal matrix $D$ 
such that $D A D^{-1}$ is balanced. The main motivation behind
this problem is that balancing a matrix does not affect its
eigenvalues, and balancing matrices in the $L_2$ norm increases
the numerical stability of eigenvalue computations~\cite{Osborne, kressner}.
Balancing also has a positive impact on the computational time
needed for computing eigenvalues~(\cite[section 1.4.3]{kressner}).
In practice, it is sufficient to get a good approximation to the
balancing problem. For $\alpha\ge 1$, a matrix $B = D A D^{-1}$ 
is an $\alpha$-approximation to the problem of balancing $A$ if 
and only if for all $i$, the ratio between the maximum and
minimum of the norms of the $i$-th row and column is 
bounded by $\alpha$. It is desirable to achieve $\alpha = 1+\eps$
for some small $\eps > 0$. A matrix $B$ that satisfies this
relaxed balancing condition is also said to be strictly 
$\eps$-balanced.

Osborne's iteration attempts to compute the diagonal matrix
$D$ by repeatedly choosing an index $i$ and balancing the
$i$-th row and column (this multiplies the $i$-th diagonal entry
of $D$ appropriately).
Osborne proposed this iteration in the
$L_2$ norm, and suggested round-robin choice of index to
balance.
However, other papers consider the iteration in other norms
and propose alternative choices of index to 
balance~\cite{parlett, schulman, ORY16}. 
Notice that a change of norm not only changes the target 
balance condition, but also changes the iteration itself,
as in each step a row-column pair is balanced in the given 
norm.
An implementation of Osborne's iteration is used in most 
numerical algebra packages, including MATLAB, LAPACK,
and EISPACK, and is empirically efficient
(see~\cite{kressner, spectra2005} for further background).
The main theoretical question about Osborne's iteration is 
its rate of convergence. How many rounds of the iteration 
are provably sufficient to get a strictly $\eps$-balanced matrix?

\paragraph{Our results.}
We consider Osborne's iteration in $L_p$ norms for finite $p$.
We design a new simple choice of the iteration (i.e., a rule
to choose the next index to balance), and we prove that this
variant provides a polynomial time approximation
scheme to the balancing problem. More specifically, we show
that in the $L_1$ norm, our implementation converges to a
strictly $\eps$-balanced matrix in 
$O\left(\eps^{-2}n^{9}\log(wn/\eps)\log w/\log n\right)$
iterations, where $\log w$ is a lower bound on the number
of bits required to represent the entries of $A$ (exact
definitions await Section~\ref{sec:preliminaries}).
%the number of bits
%required to represent the entries of $A$, assuming that $A$
%is scaled so that the minimum non-zero entry has absolute
%value $1$. 
The time complexity of these iterations is
$O\left(\eps^{-2}n^{10}\log(wn/\eps)\log w\right)$ 
arithmetic operations over $O(n\log w)$-bit numbers.
This result implies similar bounds for any $L_p$ norm where
$p$ is fixed, and in particular the important case of $p=2$.
This is because applying Osborne's iteration in the
$L_p$ norm to $A = (a_{ij})_{n\times n}$ is equivalent
to applying the iteration in the $L_1$ norm to $(a_{ij}^p)_{n\times n}$.
Of course, the bit representation complexity of the matrix, and
thus the bound on the number of iterations, grows by a factor
of $p$.

Our results give the first theoretical analysis that indicates that
Osborne's iteration in the $L_2$ norm, or any $L_p$ norm for
finite $p$, is indeed efficient in the worst case. This resolves the 
question that has been open since 1960. Previously, such a result 
was obtained only for the $L_\infty$ norm~\cite{schulman}. 
Concerning the convergence rate for the $L_p$ norms discussed 
here, we recently published a result~\cite{ORY16} that considers 
a much weaker notion of approximation. The previous result 
only shows the rate of
convergence to a matrix that is approximately balanced in an 
average sense. The matrix might still have row-column pairs that 
are highly unbalanced. The implementations in the common
numerical linear algebra packages use as a stopping condition
the strict notion of balancing, and not this weaker notion.
We discuss previous work in greater detail 
below.

\paragraph{Previous work.}
Osborne~\cite{Osborne} studied the $L_2$ norm version of matrix 
balancing, proved the uniqueness of the $L_2$ solution,
designed the iterative algorithm discussed above, and
proved that it converges in the limit to a balanced matrix
(without bounding the convergence rate). 
Parlett and Reinsch~\cite{parlett} generalized Osborne's iteration 
to other norms. Their implementation is the one widely used
in practice (see Chen~\cite[Section 3.1]{chenThesis}, also the
book~\cite[Chapter 11]{Numerical} and the
code in~\cite{EISPACK}). Grad~\cite{grad} proved convergence 
in the limit for the $L_1$ version (again without bounding the
running time), and Hartfiel~\cite{hartfiel} showed that the $L_1$ solution
is unique. Eaves et al.~\cite{eaves} gave a characterization of
balanceable matrices. Kalantari et al.~\cite{khachiyan} gave an
algorithm for $\eps$-balancing a matrix in the $L_1$ norm.
The algorithm reduces the problem to unconstrained convex
optimization and uses the ellipsoid algorithm to approximate
the optimal solution. This generates a weakly $\eps$-balanced
matrix, which satisfies the following definition. Given $\eps>0$, 
a matrix $A = (a_{ij})_{n\times n}$ is weakly $\eps$-balanced if 
and only if
$\sqrt{\sum_{i=1}^n (\|a_{.,i}\| - \|a_{i,.}\|)^2} \le \eps\cdot \sum_{i,j} |a_{i,j}|$.
%Here $a_{.,i}$ and $a_{i,.}$ denote the $i$-th column and $i$-th row respectively.
%I think this is intuitively obvious at this point - YR
Compare this with the stronger condition of being strictly 
$\eps$-balanced, which we use in this paper, and numerical 
linear algebra packages use as a stopping condition.
This condition requires that for every $i\in\{1,2,\dots,n\}$,
$\max\{\|a_{.,i}\|,\|a_{i,.}\|\}\le (1+\eps)\cdot \min\{\|a_{.,i}\|,\|a_{i,.}\|\}$.
In $L_\infty$, Schneider and Schneider~\cite{schneider} gave a
polynomial time algorithm that exactly balances a matrix. 
The algorithm does not use Osborne's iteration. Its running
time was improved by Young et al.~\cite{young91fasterparametric}. Both 
algorithms rely on iterating over computing a minimum mean cycle 
in a weighted strongly connected digraph, then contracting the cycle.
Schulman and Sinclair~\cite{schulman}
were the first to provide a quantitative bound on the running time
of Osborne's iteration. They proposed a carefully designed 
implementation of Osborne's iteration in the $L_\infty$ norm 
that strictly $\eps$-balances an $n\times n$ matrix $A$ in 
$O(n^3 \log(\varrho n/\eps))$ iterations, where $\varrho$ measures 
the initial $L_\infty$ imbalance of $A$. Their proof is an intricate 
case analysis. Finally, in~\cite{ORY16} we recently proved
that a logarithmic dependence on $1/\eps$ is impossible in the
$L_1$ norm (the lower bound is $\Omega(1/\sqrt{\eps})$). In the
same paper we also showed that several implementations of Osborne's 
iteration in $L_p$ norms, including the original implementation, converge 
to a weakly $\eps$-balanced matrix in polynomial time (which, in fact,
can be either nearly linear in $n$ or nearly linear in $1/\eps$). The
result of~\cite{ORY16} is derived by observing that Osborne's iteration can
be interpreted as an implementation of coordinate descent to
optimize the convex function from~\cite{khachiyan}. This is the starting
point of this paper, but to make the approach guarantee strict balancing,
we need to revise substantially previous implementations using novel 
algorithmic ideas. The main difficulty is the need to handle the different 
scales of row/column norm values; an index may shift between scales 
over time as a side-effect of balancing other indices. Moreover, the
analysis of the convergence rate is more complicated, and requires
additional ideas.

\paragraph{Our contribution.}
The main difficulty with respect to previous work is the following. The 
convergence rate of coordinate descent can be bounded effectively 
as long as there is a choice of coordinate (i.e., index) for which the 
drop in the objective function in a single step is non-negligible compared 
with the current objective value. But if this is not the case, then one 
can argue only about the balance of each index relative to the sum 
of norms of all rows and columns. Indices that have relatively heavy 
weight (row norm $+$ column norm) will indeed be balanced at this 
point. However, light-weight indices may be highly unbalanced. The
naive remedy to this problem is to work down by scales. After
balancing the matrix globally, heavy-weight indices are balanced,
approximately, so they can be left alone, deactivated. Now there 
are light-weight indices that have become heavy-weight with respect 
to the remaining active nodes, so we can continue balancing the
active indices until the relatively heavy-weight among them become
approximately balanced, and so forth. The problem with the naive
solution is that balancing the active indices shifts the weights of both
active and inactive indices, and they move out of their initial scale.
If the scale sets of indices keep changing, it is hard to argue that
the process converges. Shifting between scales is precisely what 
our algorithm and proof deal with. Light-weight indices that have
become heavy-weight are easy to handle. They can keep being
active. Heavy-weight indices that have become light-weight cannot
continue to be inactive, because they are no longer guaranteed to
be approximately balanced. Thus, in order to analyze convergence
effectively, we need to bound the number (and global effect on weight)
of these reactivation events.

\section{Preliminaries}\label{sec:preliminaries}

The input is a real square matrix $A = (a_{ij})_{n\times n}$.
We denote the $i$-th row of such a matrix by $a_{i,.}$ and
the $i$-th column by $a_{.,i}$. We also use the notation
$[n] = \{1,2,\dots,n\}$. The matrix $A$ is {\em balanced} 
in the $L_p$ norm iff $\|a_{.,i}\|_p = \|a_{i,.}\|_p$ for every 
index $i\in[n]$. Since the condition for being balanced
depends neither on the signs of the entries of $A$ nor
on the diagonal values, we will assume without loss of
generality that $A$ is non-negative with zeroes on the
diagonal.

An invertible diagonal matrix $D=\diag(d_1, \cdots, d_n)$ 
{\em balances} $A$ in the $L_p$ norm iff $DAD^{-1}$ is 
balanced in the $L_p$ norm. A matrix $A$ is {\em balanceable} 
iff there exists an invertible diagonal matrix $D$ that balances 
$A$. Balancing a matrix $A=(a_{ij})_{n\times n}$ in the $L_p$ 
norm is equivalent to balancing the matrix $(a_{ij}^p)_{n\times n}$
in the $L_1$ norm. Therefore, for the rest of the paper we focus
on balancing matrices in the $L_1$ norm.

We use $a_{\min}$ to denote the minimum non-zero entry of 
$A$. We also define $w = \frac{1}{a_{\min}}\cdot\sum_{ij} a_{ij}$.
\begin{definition}\label{def:balanced}
Given $\eps > 0$ and an $n\times n$ matrix $A$,
we say that the index $i$ of $A$ (where $i\in [n]$) 
is {\em $\eps$-balanced} iff
$$
\frac{\max\left\{\ali{i},\alo{i}\right\}}{\min\left\{\ali{i},\alo{i}\right\}}\le 1 + \eps.
$$
We say that $A$ is {\em strictly $\eps$-balanced} iff every index 
$i$ of $A$ is $\eps$-balanced.
\end{definition}
 
Any implementation of Osborne's iteration can be thought of as 
computing vectors $\xt{t}\in \mathbb{R}^n$ for $t=1,2,\dots$, where 
iteration $t$ is applied to the matrix $(a_{ij}^{(t)}) = D A D^{-1}$ for 
$D = \diag(e^{x_1^{(t)}},e^{x_2^{(t)}},\dots,e^{x_n^{(t)}})$.
Thus, for all $i,j$, $a_{ij}^{(t)} = a_{ij}\cdot e^{x_i^{(t)} - x_j^{(t)}}$.
Initially, $\x^{(1)} = (0, 0,\dots, 0)$. A balancing step of the iteration
chooses an index $i$, then sets 
$x_i^{(t+1)} = x_i^{(t)} + \frac 1 2\cdot\left(\ln\|a_{.,i}^{(t)}\|_1-\ln\|a_{i,.}^{(t)}\|_1\right)$,
and for all $j\ne i$, keeps $x_j^{(t+1)} = x_j^{(t)}$.
For $\x\in\mathbb{R}^n$, we denote the sum of entries of the
matrix $D A D^{-1}$ for $D = \diag(e^{x_1},e^{x_2},\dots,e^{x_n})$
by $f(\x) = f_A(\x) = \sum_{ij} a_{ij}\cdot e^{x_i - x_j}$.
For any $n\times n$ non-negative matrix $B = (b_{ij})$, we denote 
by $G_B$ the weighted directed graph with node set $\{1,2,\dots,n\}$,
arc set $\{(i,j): b_{ij} > 0\}$, where an arc $(i,j)$ has weight $b_{ij}$.
We will assume henceforth that the undirected version of $G_A$
is connected, otherwise we can handle each connected component
separately.
We quote a few useful lemmas. The references contain the proofs.
\begin{lemma}[Theorem 1 in Kalantari et al.~\cite{khachiyan}]\label{lm: bal-opt}
The input matrix $A$ is balanceable if and only if $G_A$ is strongly 
connected. Moreover, $D A D^{-1}$ is balanced in the $L_1$ norm
if and only if $D = \diag(e^{x_1^*},e^{x_2^*},\dots,e^{x_n^*})$,
where $\x^* = (x_1^*,x_2^*,\dots,x_n^*)$ minimizes
$f(\x)$ over $\x\in\mathbb{R}^n$.
\end{lemma}

Notice that $f$ is a convex function and 
the gradient $\nabla f(\x)$ of $f$ at $\x$ is given by
$$
\frac{\partial f(\x)}{\partial x_i} = 
\sum_{j=1}^n a_{ij}\cdot e^{x_i - x_j} - \sum_{j=1}^n a_{ji}\cdot e^{x_j - x_i},
$$
the difference between the total weight of arcs leaving node $i$
and the total weights of arcs going into node $i$ in the graph of
$D A D^{-1}$ for $D = \diag(e^{x_1},e^{x_2},\dots,e^{x_n})$.
If $D A D^{-1}$ is balanced then 
the arc weights $a_{ij}\cdot e^{x_i - x_j}$ form a valid
circulation in the graph $G_A$, since the gradient has to
be $0$. Some properties of $f$ are given in the following
lemma.
\begin{lemma}[Lemmas 2.1 and 2.2 in Ostrovsky et al.~\cite{ORY16}]\label{lemma:opt lower}
If $\x'$ is derived from $\x$ by balancing index $i$ of a
matrix $B = (b_{ij})_{n\times n}$, then 
$f(\x) - f(\x') = (\sqrt{\|b_{.,i}\|_1}-\sqrt{\|b_{i,.}\|_1})^2$.
Also, for all $\x\in\mathbb{R}^n$, 
$f(\x)-f(\x^*) \le \frac{n}{2}\cdot\|\nabla f(\x)\|_1$.
\end{lemma}

We also need the following absolute bounds on the arc weights.
\begin{lemma}[Lemma 3.2 in Ostrovsky et al.~\cite{ORY16}]\label{lm: bounds on weight}
Suppose that a matrix $B$ is derived from a matrix $A$ through
a sequence of balancing operations. Then, for every arc $(i,j)$
of $G_B$, $\left(\frac{a_{\min}}{\sum_{ij} a_{ij}}\right)^n\cdot
\sum_{ij} a_{ij}\le b_{ij}\le\sum_{ij} a_{ij}$. (Notice that the arcs
of $G_B$ are identical to the arcs of $G_A$.)
\end{lemma}

Finally, we prove the following global condition on indices being
$\eps$-balanced.
\begin{lemma}\label{lemma:weight}
Consider a matrix $B = D A D^{-1} = (b_{ij})_{n\times n}$, 
where $D = \diag(e^{x_1},e^{x_2},\dots,e^{x_n})$,
that was derived from $A$ by a sequence of zero or more
balancing operations.
Let $\eps\in(0,1/2]$, and put $\eps'=\frac{\eps^2}{64n^4}$. Suppose
that $\|\nabla f_A(\vec{0})\|_1 \le \eps' \cdot f_A(\vec{0})$. 
Then, for every $i\in[n]$ we have the following implication.
If $\|b_{.,i}\|_1 + \|b_{i,.}\|_1 \ge \frac{1}{8n^3}\cdot f_A(\x)$, 
then index $i$ is $\eps$-balanced in $B$.
\end{lemma}

\begin{proof}
We will show the contrapositive claim that if a node is 
not $\eps$-balanced then it must have low weight
(both with respect to $B$).
Let $i$ be an index that is not $\eps$-balanced in $B$.
Without loss of generality we may assume that the in-weight
is larger than the out-weight, so ${\|b_{.,i}\|_1}/{\|b_{i,.}\|_1} > 1+\eps$.
Consider what would happen if we balance index $i$ in $B$,
yielding a vector $\x'$ that differs from $\x$ only in the $i$-th
coordinate.
\begin{eqnarray}
\nonumber
            f_A(\x)-f_A(\x') 
& = & \left(\sqrt{\|b_{.,i}\|_1} - \sqrt{\|b_{i,.}\|_1}\right)^2 \\
\nonumber
& > & \|b_{.,i}\|_1\cdot\left(1-\sqrt{\frac{1}{1+\eps}}\right)^2 \\
\label{eq: f lower}
& > & \frac{\eps^2}{16}\cdot \left(\|b_{.,i}\|_1 + \|b_{i,.}\|_1\right),
\end{eqnarray}
where the equation follows from Lemma~\ref{lemma:opt lower}
and the last inequality uses the fact that $\eps\le\frac 1 2$.
%
%NOTE: Here is a proof for constant of (1)
%
%$$
%\left(1-\sqrt{\frac{1}{1+\eps}}\right)^2 = \frac{\left(1-\frac{1}{1+\eps}\right)^2}{\left(1+\sqrt{\frac{1}{1+\eps}}\right)^2}=
%\frac{\frac{\eps^2}{(1+\eps)^2}}{{\left(1+\sqrt{\frac{1}{1+\eps}}\right)^2}} = \frac{\eps^2}{(1+\eps)^2\left(1+\sqrt{\frac{1}{1+\eps}}\right)^2} = \frac{\eps^2}{[(1+\eps) + \sqrt{1+\eps}]^2} \ge \frac{\eps^2}{16}
%$$
%
%The last inequality holds because $[(1+\eps) + \sqrt{1+\eps}]\le 3$ for $\eps <1/2$.\\
%
On the other hand, we have
\begin{eqnarray}
\nonumber
            f_A(\x) - f_A(\x') 
& \le &  f_A(\vec{0}) - f(\x^*) \\
\nonumber
& \le & \frac{n}{2}\cdot\|\nabla f_A(\vec{0})\|_1 \\
\nonumber
& \le & \frac{n}{2}\cdot\eps'\cdot f_A(\vec{0}) \\
\label{eq: f upper}
& = & \frac{\eps^2}{128n^3}\cdot f_A(\vec{0}). 
\end{eqnarray}
where the first inequality follows from the
the fact that every balancing step decreases
$f_A$, the second inequality follows from
Lemma~\ref{lemma:opt lower}, the third
inequality follows from the assumption on
$f_A(\vec{0})$, and the last equation follows
from the choice of $\eps'$. 
Combining the bounds on $f_A(\x) - f_A(\x')$ in
Equations~\eqref{eq: f lower} and~\eqref{eq: f upper}
gives
$$
\|b_{.,i}\|_1 + \|b_{i,.}\|_1 < \frac{1}{8n^3}\cdot f_A(\vec{0}),
$$
and this completes the proof.
\end{proof}

%%%%%%%%%%%%%%%%%%%%%%%%%%%%%%%%%%%%
%%%%%   Strict Balancing
%%%%%%%%%%%%%%%%%%%%%%%%%%%%%%%%%%%%
\section{Strict Balancing}\label{section:strict balancing}

\begin{algorithm}[t!]
\caption{StrictBalance($A$, $\eps$)}\label{alg:strict}
\begin{algorithmic}[1]
\Input{Matrix $A \in\mathbb{R}^{n\times n}, \eps$}
\Output{A strictly $\eps$-balanced matrix}\vspace{2mm}
\State $\mathcal{B}_1=\varnothing$, $\tau_1 = 0$, $s = 1$, 
          $\eps'=\eps^2/64n^4$, $\xt{1}=(0,\ldots,0)$, $t =1$\vspace{2mm}
\While{$\mathcal{B}_s\neq [n]$ and there is $i\in[n]$ that is not $\eps$-balanced}\vspace{2mm}
\State Define $\fbk{s}:\mathbb{R}^n\rightarrow \mathbb{R}$, 
               $\fbk{s}(\x) = \displaystyle \sum_{i,j: i\notin{\cal B}_s \text{ or } j\notin{\cal B}_s}
                       a_{ij}e^{x_i-x_j}$ \vspace{2mm}
\While{$\displaystyle\frac{\|\nabla \fbk{s}(\xt{t})\|_1}{\fbk{s}(\xt{t})} > \eps'$}
      \State Pick $i=\argmax_{i\notin\mathcal{B}_s}\left\{\left(\sqrt{\|a_{.,i}^{(t)}\|_1} - \sqrt{\|a_{i,.}^{(t)}\|_1}\right)^2\right\}$
      \State Balance $i$th node: $\xt{t+1} = \xt{t} + \alpha_t\mathbf{e}_i$, where $\alpha_t = \ln\sqrt{\|a_{.,i}^{(t)}\|_1/{\|a_{i,.}^{(t)}\|_1}}$\vspace{2mm}
      \State $t \leftarrow t +1$\vspace{2mm}
      \If {$s > 1$ and $\|a_{.,i}^{(t)}\|_1 + \|a_{i,.}^{(t)}\|_1 < \tau_s$ for some $i\in\B_s \setminus \B_{s-1}$}\vspace{2mm}
          \State $\mathcal{B}_s = \mathcal{B}_s\setminus \{i\notin\mathcal{B}_{s-1}: \|a_{.,i}^{(t)}\|_1 + \|a_{i,.}^{(t)}\|_1 < \tau_s \}$\vspace{2mm}
          \State Redefine $\fbk{s}:\mathbb{R}^n\rightarrow \mathbb{R}$, $\fbk{s}(\x) = \displaystyle\sum_{i,j: i\notin\mathcal{B}_s\text{ or }j\notin\mathcal{B}_s}a_{ij}e^{x_i-x_j}$\vspace{2mm}
      \EndIf
 \EndWhile\vspace{2mm}
 \State $\tau_{s+1} = \displaystyle\frac{1}{4n^3}\fbk{s}(\xt{t})$\vspace{2mm}
 \State $\mathcal{B}_{s+1}=\mathcal{B}_s\cup\Big\{i:\|a_{.,i}^{(t)}\|_1 + \|a_{i,.}^{(t)}\|_1 \ge \tau_{s+1}\Big\}$\vspace{2mm}
 \State $s \leftarrow s + 1$\vspace{2mm}
\EndWhile
\State\textbf{return} the resulting matrix
\end{algorithmic}
\end{algorithm}

In this section we present a variant of Osborne's iteration
and prove that it converges in polynomial time to a strictly 
$\eps$-balanced matrix. The algorithm, a procedure named
StrictBalance, is defined in pseudocode labeled 
Algorithm~\ref{alg:strict} on page~\pageref{alg:strict}.
Lemma~\ref{lemma:weight} above motivates the main idea 
of contracting heavy nodes in step 14 of StrictBalance.

Our main theorem is
\begin{theorem}\label{thm:strict}
StrictBalance($A$, $\eps$) returns a strictly $\eps$-balanced matrix 
$B = D A D^{-1}$ after at most 
$$
O\left(\eps^{-2}n^{9}\log(wn/\eps)\log w/\log n\right)
$$
balancing steps, using $O\left(\eps^{-2}n^{10}\log(wn/\eps)\log w\right)$ 
arithmetic operations over $O(n\log w)$-bit numbers.
\end{theorem}

The proof of Theorem~\ref{thm:strict} uses a few arguments, given
in the following lemmas. A {\em phase} of StrictBalance is one iteration of
the outer while loop. Notice that in the beginning of this loop the
variable $s$ indexes the phase number (i.e., $s-1$ phases
were completed thus far). Also in the beginning of the inner while
loop the variable $t$ indexes the total iteration number 
from all phases (i.e., $t-1$ balancing operations from all
phases were completed thus far). 

We identify outer loop iteration $s$ with an interval
$[t_s,t_{s+1}) = \{t_s,t_s+1,\dots,t_{s+1}-1\}$ of the 
inner loop iterations executed during
phase $s$. We denote by ${\cal B}_{s,t}$ the value of ${\cal B}_s$
in the beginning of the inner while loop iteration number $t$
(dubbed time $t$). If $t\in[t_j,t_{j+1})$, then ${\cal B}_{s,t}$ is 
defined only for $s\le j$. We also use $\gbk{s,t}$ to denote the
graph that is obtained by contracting the nodes of set ${\cal B}_{s,t}$
in graph $G_A$. Also $\fbk{s,t}$ is the function corresponding to graph $\gbk{s,t}$
and $\fbk{s,t}(\xt{t})$ denotes the sum of weights of arcs of graph $\gbk{s,t}$ at
time $t$. If set $B_s$ is unchanged during an interval and there is
no confusion, we may use $\gbk{s}$ instead of $\gbk{s,t}$. Particularly
we use $\fbk{s}(\xt{t})$ instead of $\fbk{s,t}(\xt{t})$.
We refer to the quantity 
$\|a_{.,i}^{(t)}\|_1 + \|a_{i,.}^{(t)}\|_1$ as the {\em weight} of
node $i$ at time $t$.

\begin{lemma}\label{lm: B_s at time t}
For every phase $s \ge 1$, for every $t\ge t_{s+1}$, 
$\B_{s,t} = \B_{s,t_{s+1}}$.
\end{lemma}

\begin{proof}
The claim follows easily from the fact that 
any iteration $t\ge t_{s+1}$ belongs to a phase 
$s' > s$, so 
$\B_{s,t_{s+1}}\cap (\B_{s',t}\setminus \B_{s'-1,t}) = \emptyset$,
and by line 8 and 9 of StrictBalance none of the nodes in
$\B_{s,t_{s+1}}$ will be removed.
\end{proof} 

\begin{lemma}\label{lemma:weight_drop}
For all $s > 1$, for all $t\in [t_s,t_{s+1})$, 
$\fbk{s,t}(\xt{t}) \le \left(n -  \left|{\B}_{s,t}\right|\right)\cdot\tau_{s}$.
\end{lemma}

\begin{proof}
Let $t_s = t_{s,1} < t_{s,2} < t_{s,3} < \cdots < t_{s,\ell_s}$ denote
the time steps before which $\B_s$ changes during phase $s$.
For simplicity, we abuse notation and use $\B_{s,j}$ instead of
$\B_{s,t_{s,j}}$. Clearly 
$\B_{s,1}\supseteq \B_{s,2}\ldots\supseteq\B_{s,\ell_s}$, because 
we only remove nodes from $\B_{s}$ once it is set. Fix $s > 1$. 
We prove this lemma by induction on $r\in\{1,2,\dots,\ell_s\}$.
For the basis, let $r = 1$.
Clearly, by the way the algorithm sets $\B_s$ before time $t_{s,1}$,
all nodes with weight $\ge \tau_s$ are in $\B_s$, and therefore
every node $i\not\in\B_s$ has weight at most $\tau_s$, so the
lemma follows. Now, assume that the lemma is true for every
$t \le t_{s,r}$, we show that the lemma holds for every
$t \le t_{s,r+1}$. If $t\in [t_{s,r}, t_{s,r+1})$, then $\B_{s,t} = \B_{s,t_{s,r}}$,
and we have:
%We use $\B_s$ to denote both quantities: 
$$
\fbk{s}(\xt{t}) \le
\fbk{s}(\xt{t_{s,r}})\le\left(n -  \left|{\B}_{s,t_{s,r}}\right|\right)\cdot\tau_{s} = 
\left(n -  \left|{\B}_{s,t}\right|\right)\cdot\tau_{s}.
$$
The first inequality holds because balancing operations from time $t_{s,r}$ to
time $t$ only reduce the value of $\fbk{s}$, and the second inequality holds
by the induction hypothesis.

Just before iteration $t = t_{s,r+1}$, the set $\B_s$ changes,
and one or more nodes are removed from it.
However, every removed node has weight at most $\tau_s$, and
its removal does not change the weights of the other nodes in
$[n]\setminus \B_s$. Therefore, if $k$ nodes are removed from
$\B_s$,
$$
\fbk{s}(\xt{t_{s,r+1}})\le \left(n -  \left|{\B}_{s,t_{s,r}}\right|\right)\cdot\tau_{s}
+ k\cdot \tau_{s} = \left(n -  \left|{\B}_{s,t_{s,r+1}}\right|\right)\cdot\tau_{s}.
$$
This completes the proof.
\end{proof}

\begin{corollary}\label{cor: tau}
For all $s > 1$, $\fbk{s}(\xt{t_{s+1}}) \le 
\frac{1}{4n^2}\cdot \fbk{s-1}(\xt{t_s})$. If $s > 2$,
then $\tau_{s} \le \frac{\tau_{s-1}}{4n^2}$.
\end{corollary}

\begin{proof}
Notice that
$$
\fbk{s}(\xt{t_{s+1}}) \le
n\cdot\tau_s = \frac{1}{4n^2}\cdot \fbk{s-1}(\xt{t_s}),
$$
where the inequality follows from Lemma~\ref{lemma:weight_drop},
and the equation follows from line 13 of StrictBalance. This proves 
the first assertion. As for the second assertion, notice that if $s > 2$
then $s-1 > 1$, so using line 13 of StrictBalance and 
Lemma~\ref{lemma:weight_drop} again,
$$
\tau_s = \frac{1}{4n^3}\cdot \fbk{s-1}(\xt{t_s})\le
\frac{1}{4n^3}\cdot n\tau_{s-1} = \frac{1}{4n^2}\cdot \tau_{s-1},
$$
as stipulated.
\end{proof}

\begin{lemma}\label{lm: B_s balanced}
For every phase $s > 1$, for every $t \ge t_s$,
all the nodes in $\B_{s,t}$ have weight $\ge \tau_s/2$ and are $\eps$-balanced 
at time $t$.
\end{lemma}

\begin{proof}
Fix $s > 1$ and let $i\in \B_{s,t}$. Without loss of generality
$i\not\in \B_{s-1,t}$, otherwise we can replace $s$ with $s-1$.
(Recall that $\B_1 = \emptyset$ at all times.)
Also note that it must be the case that $i\in \B_{s,t_s}$, because
$\B_s$ does not accumulate additional nodes after being created.
If $t\in [t_s,t_{s+1}]$, then lines 13-14 and 8-9 of
StrictBalance guarantee that if $i\in \B_{s,t}\setminus \B_{s-1,t}$,
then its weight at time $t$ is at least $\tau_{s}$.

Otherwise, consider $t > t_{s+1}$ and let $s' > s$ be the phase 
containing $t$. Consider a phase $j > s$.
By Lemma~\ref{lemma:weight_drop} the total
weight of $\fbk{j}$ during phase $j$ is at most $n\tau_j$, and $\fbk{j}$ never drops below
0. So, the total weight that a node $i\in\B_j$ can lose (which is at most the total
weight that $\fbk{j}$ can lose) is at most $n\tau_j$.
By Corollary~\ref{cor: tau}, for every $j > s$,
$\tau_{j+1}\le\frac{\tau_j}{4n^2}$. Now, suppose that $t$ is an iteration
in phase $s' > s$. Then, the weight of $i$ at time $t$ is at least
$$
\tau_s - \sum_{j=s+1}^{s'} n\tau_j \ge
\tau_s\cdot\left(1 - n\cdot\sum_{k=1}^{s'-s} (2n)^{-2k}\right) \ge \frac{\tau_s}{2}.
$$
Thus we have established that at any time $t\ge t_s$, if
$i\in\B_{s,t}$ then its weight is at least 
$\frac{\tau_s}{2}  = \frac{1}{8n^3}\fbk{s-1}(\xt{t_s})$.
By line 4 of StrictBalance, 
$\|\nabla \fbk{s-1,t_s}(\xt{t_s})\|_1\le \eps'\cdot \fbk{s-1,t_s}(\xt{t_s})$.
By Lemma~\ref{lm: B_s at time t},
$\B_{s-1}$ does not change in the interval $[t_s,t]$.
Therefore, we conclude from Lemma~\ref{lemma:weight}
that $i$ is $\eps$-balanced at time $t$.
\end{proof}

\begin{lemma}\label{lemma:phase}
Suppose that $t < t'$ satisfies $[t,t')\subseteq [t_s,t_{s+1})$, and
furthermore, during the iterations in the interval $[t,t')$ the
set ${\cal B}_s$ does not change (it could change after
balancing step $t'-1$). Then, the length of the interval
$$
t' - t = O\left(\eps^{-2}n^7\log(wn/\eps)\right).
$$
\end{lemma}

\begin{proof}
Rename the nodes so that $\B_{s,t} = \B_{s,t'-1} = \{p, p+1, \ldots, n\}$.
The assumption that $\B_s$ does not change during the interval
$[t,t')$ means that the weights of all the nodes $p,p+1,\dots,n$
remain at least $\tau_s$ for the duration of this interval.
During the interval $[t,t')$, the graph $\gbk{s}$ (which remains fixed)
is obtained by contracting the nodes $p, p+1, \ldots, n$ in $G_A$.
So $\gbk{s}$ has $p$ nodes $1,2,\dots,p-1,p$, where the last node
$p$ is the contracted node. In each iteration in the interval $[t,t')$,
one of the nodes $1,2,\dots,p-1$ is balanced. Consider some time
step $t''\in [t,t')$, and let $I_i$ and $O_i$, respectively, denote the
current sums of weights of the arcs of $\gbk{s}$ into and out of node 
$i$, respectively. Let $j\in[p-1]$ be the node that 
maximizes $\frac{(I_j - O_j)^2}{I_j + O_j}$. We have
\begin{eqnarray}\label{eq:reduce}
\fbk{s}(\xt{t''}) - \fbk{s}(\xt{t''+1}) & = &\max_{i\in[p-1]} \left(\sqrt{I_i} - \sqrt{O_i}\right)^2 \ge  \left(\sqrt{I_j} - \sqrt{O_j}\right)^2
%=  \left(\frac{I_j - O_j}{\sqrt{I_j} + \sqrt{O_j}}\right)^2\nonumber\\
  \ge  \frac{(I_j - O_j)^2}{2(I_j + O_j)}\nonumber\\
 &\ge& \frac{\sum_{i=1}^{p-1}(I_i - O_i)^2}{2\sum_{i=1}^{p-1}(I_i + O_i)} 
\ge \frac{\left(\sum_{i=1}^{p-1} |I_i - O_i|\right)^2}{2n\sum_{i=1}^{p}(I_i + O_i)} 
 \ge  \frac{\left(\sum_{i=1}^{p} |I_i - O_i|\right)^2}{8n\sum_{i=1}^{p}(I_i + O_i)} \nonumber\\
& = & \frac{1}{16n}\cdot\frac{\|\nabla \fbk{s}(\xt{t''})\|_1^2}{\fbk{s}(\xt{t''})}.
\end{eqnarray}
The first equation follows from the choice of $i$ in line 5 
StrictBalance, and Lemma~\ref{lemma:opt lower}.
The third inequality follows from an averaging argument and the choice of $j$.
The fourth inequality uses Cauchy-Schwarz. The last inequality holds 
because $\sum_{i=1}^{p} (I_i - O_i) = 0$, so 
$|I_p - O_p|= \left|\sum_{i=1}^{p-1} (I_i - O_i) \right| \le 
\sum_{i=1}^{p-1} |I_i - O_i|$, and therefore
$\sum_{i=1}^{p} |I_i - O_i| \le 2\sum_{i=1}^{p-1} |I_i - O_i|$.

Since the interval $[t,t')$ is contained in phase $s$, the stopping
condition for the phase does not hold, so
$$
\frac{\|\nabla \fbk{s}(\xt{t''})\|_1}{\fbk{s}(\xt{t''})} > \eps' = \frac{\eps^2}{64n^4}.
$$ 
Therefore, 
\begin{eqnarray*}
\fbk{s}(\xt{t''})-\fbk{s}(\xt{t''+1)}) &\ge&
     \frac{1}{16n}\cdot\frac{\|\nabla \fbk{s}(\xt{t''}\|_1^2}{\fbk{s}(\xt{t''})} \\
& > & \frac{\eps'}{16n}\cdot\|\nabla \fbk{s}(\xt{t''})\|_1 \\
&\ge& \frac{\eps'}{8n^2}\cdot(\fbk{s}(\xt{t''}) - \fbk{s}(\x^*)),
\end{eqnarray*}
where the last inequality follows from Lemma~\ref{lemma:opt lower}. 
Rearranging the terms gives
$$
\fbk{s}(\xt{t''+1})-\fbk{s}(\x^*)\le 
\left(1-\frac{\eps'}{8n^2}\right)\cdot (\fbk{s}(\xt{t''}) - \fbk{s}(\x^*)).
$$
Iterating for $T$ step yields
$$
\fbk{s}(\xt{t+T})-\fbk{s}(\x^*)\le 
\left(1-\frac{\eps'}{8n^2}\right)^T\cdot (\fbk{s}(\xt{t}) - \fbk{s}(\x^*)).
$$
Now, by Lemma~\ref{lm: bounds on weight}, we have that
$\fbk{s}(\xt{t}) - \fbk{s}(\x^*)\le \fbk{s}(\xt{t}) \le\sum_{i,j=1}^n a_{ij}$, 
and for all $t''$, $\fbk{s}(\xt{t''})\ge \frac{1}{w^n} \sum_{i,j=1}^n a_{ij}$.
Therefore, if
$t'-t\ge\frac{8n^2}{\eps'}\cdot\ln\left(16 n w^n / (\eps')^2 \right) + 1$,
then 
$$
\fbk{s}(\xt{t'-1})-\fbk{s}(\x^*) \le 
\left(\frac{\eps'}{4\sqrt{n}}\right)^2\cdot \frac{1}{w^n}\cdot \sum_{i,j=1}^n a_{ij} 
\le \left(\frac{\eps'}{4\sqrt{n}}\right)^2 \cdot \fbk{s}(\xt{t'-1}).
$$ 
Therefore,
$$
\frac{1}{16n}\cdot \frac{\|\nabla \fbk{s}(\xt{t'-1}\|_1^2}{(\fbk{s}(\xt{t'-1}))^2} \le 
\frac{\fbk{s}(\xt{t'-1}) -\fbk{s}(\xt{t'})}{\fb(\xt{t'-1})} \le 
\frac{\fbk{s}(\xt{t'-1}) - \fbk{s}(\x^{*})}{\fb(\xt{t'-1})} \le \left(\frac{\eps'}{4\sqrt{n}}\right)^2,
$$
where the first inequality follows from~(\ref{eq:reduce}), and the second inequality
holds because $\fbk{s}(\x^{*}) \le \fbk{s}(\xt{t'-1})$.
We get that
$\frac{\|\nabla \fb(\xt{t'-1})\|_1}{\fb(\xt{t'-1})} \le \eps'$,
in contradiction to our assumption that the phase does not end before 
the start of iteration $t'$.
\end{proof}

\begin{corollary}\label{cor:phase}
In any phase, the number of balancing steps is at most
$O\left(\eps^{-2}n^8\log(wn/\eps)\right)$.
\end{corollary}

\begin{proof}
In the beginning of phase $s$ the set $\B_s$ contains at most
$n-1$ nodes. Partition the phase into intervals $[t,t')$ where
$\B_s$ does not change during an interval, but does change
between intervals. By Lemma~\ref{lemma:phase}, each interval
consists of at most $O\left(\eps^{-2}n^7\log(wn/\eps)\right)$ 
balancing steps. Since nodes that are removed from $\B_s$
between intervals are never returned to $\B_s$, the number
of such intervals is at most $n-1$. Hence, the total number of
balancing steps in the phase is at most
$O\left(\eps^{-2}n^8\log(wn/\eps)\right)$.
\end{proof}

\begin{lemma}\label{lm: number of phases}
The total number of phases of the algorithm is $O(n\log w / \log n)$.
\end{lemma}

\begin{proof}
Let $s>2$ be a phase of the algorithm and $t\in[t_s, t_{s+1})$.
By Lemma~\ref{lm: bounds on weight},
$\fbk{s,t}(\xt{t})\ge \frac{1}{w^n}\cdot\sum_{ij} a_{ij}$.
On the other hand, by Lemma~\ref{lemma:weight_drop} and Corollary~\ref{cor: tau}, $\tau_s\le \frac{1}{(4n^2)^{s-2}}\cdot\tau_2\le
\frac{1}{(4n^2)^{s-2}}\cdot\sum_{ij} a_{ij}$, and
$\fbk{s,t}(\xt{t})\le n\tau_s$. 
Combining these gives $\frac{1}{w^n}\cdot\sum_{ij} a_{ij} \le n\tau_s \le 
\frac{n}{(4n^2)^{s-2}}\cdot\sum_{ij} a_{ij}$ which implies that
$s \le \frac{\log (nw^n)}{\log (4n^2)} + 2$.
\end{proof}

\begin{proof}[Proof of Theorem~\ref{thm:strict}]
By Lemma~\ref{lm: number of phases}, for some
$s = O(n\log w / \log n)$, StrictBalance terminates, so $\B_{s,t_s} = [n]$.
By Corollary~\ref{cor:phase}, the number of
balancing steps in a phase is at most 
$O\left(\eps^{-2}n^8\log(wn/\eps)\right)$. Therefore, the
total number of balancing steps is at most
$O\left(\eps^{-2}n^{9}\log(wn/\eps)\log w/\log n\right)$. 
These balancing steps require at most
$O\left(\eps^{-2}n^{10}\log(wn/\eps)\log w\right)$ arithmetic 
operations over $O(n\log w)$-bit numbers.
When the algorithm terminates at time $t_s$,
all the nodes are in $\B_{s,t_s}$, and by Lemma~\ref{lm: B_s balanced}
they are all $\eps$-balanced, so the matrix is strictly
$\eps$-balanced.
\end{proof}

\bibliographystyle{plain}
\bibliography{main}

\end{document}